\documentclass[twocolumn, pra , nofootinbib %superscriptaddress
]{revtex4-1}
\usepackage{graphicx}
\usepackage{amsfonts}
\usepackage{amssymb}
\usepackage{amsthm}
\usepackage{array}
\usepackage{amsmath}
\usepackage{verbatim} 
\usepackage{hyperref}
\usepackage{color}
\usepackage{bbold}
\usepackage{epstopdf}
\usepackage{mathtools}
\usepackage{enumerate}
\usepackage{subcaption}
\usepackage{color}
\usepackage{ulem}
\usepackage{physics}

\newtheorem{proposition}{Proposition}
\newtheorem*{proposition*}{Proposition}
\newtheorem{definition}{Definition}
\newtheorem{theorem}{Theorem}
\newtheorem*{theorem*}{Theorem}

\newtheorem*{corollary*}{Corollary}

\newtheorem{lemma}{Lemma}

\def\Tr{\mathrm{Tr}}

\begin{document}
\title{Negativity of quasiprobability distributions as a measure of nonclassicality}
\author{Kok Chuan Tan}
\email{bbtankc@gmail.com}
\author{Seongjeon Choi}
\author{Hyunseok Jeong}
\email{jeongh@snu.ac.kr}
\affiliation{Center for Macroscopic Quantum Control \& Institute of Applied Physics, Department of Physics and Astronomy, Seoul National University, Seoul, 08826, Korea}

\begin{abstract}
We demonstrate that the negative volume of any $s$-paramatrized quasiprobability, including the Glauber-Sudashan $P$-function, can be consistently defined and forms a continuous hierarchy of nonclassicality measures that are linear optical monotones. These measures therefore belong to an operational resource theory of nonclassicality based on linear optical operations. The negativity of the Glauber-Sudashan $P$-function in particular can be shown to have an operational interpretation as the robustness of nonclassicality. We then introduce an approximate linear optical monotone, and show that this nonclassicality quantifier is computable and is able to identify the nonclassicality of nearly all nonclassical states.
\end{abstract}

\maketitle

\section{Introduction}
It is typically considered that the most classical quantum states of a light field, or more generally, bosonic fields, are the coherent states\cite{Glauber1963}. Defined as the eigenstates of the annihilation operator, $a\ket{\alpha}=\alpha\ket{\alpha}$, the dynamics of coherent states in a quadratic potential closely resemble that of a classical harmonic oscillator\cite{Schleich2001}. The seminal work of Glauber\cite{Glauber1963} and Sudarshan\cite{Sudarshan1963} showed that every quantum state of light may be written in the form $$\rho = \int d^2 \alpha P(\alpha) \ket{\alpha}\bra{\alpha}$$ where the coefficient $P(\alpha)$ is referred to as the Glauber-Sudarshan $P$-function. When $P(\alpha)$ corresponds to a proper probability density function, the quantum state may be considered a statistical mixture of coherent states and is hence classical. More generally, $P(\alpha)$ is a quasiprobability distribution that may not correspond to any classical probability density. In such cases, the state is considered nonclassical. It is a well known fact that the only classical pure states are the coherent states\cite{Hillery1985}.

Nonclassical states find useful applications in a wide range of tasks, such as quantum metrology\cite{Caves1981}, quantum teleportation\cite{Furusawa1998}, quantum cryptography\cite{Hillery2000}, quantum communication\cite{Braunstein2005} and quantum information processing\cite{Bartlett2003}. Correspondingly, there has been great interest in the characterization, verification and quantification of nonclassicality in light. As the $P$-function function is frequently highly singular, involving terms such as the $n$th order derivatives of delta functions\cite{Agarwal2012}, it is neither theoretically nor experimentally accessible in many instances. As such, previous efforts have largely focused on finding methods to quantify nonclassicality via other means. The Mandel Q parameter\cite{Mandel1979} for instance, measures the deviation from Poissonian statistics. The entanglement potential quantifies the maximum amount of entanglement that can be generated from a beam splitter\cite{Asboth2005}. The nonclassicality depth quantifies the amount of interaction with a thermal state in order to erase nonclassicality\cite{Lee1991, Kuhn2018}. One may also count the number of superpositions of coherent states\cite{Gehrke2012}, the amount of coherent superposition between coherent states\cite{Tan2017}, the sensitivity of a quantum state to operator ordering\cite{Bievre2019}, various geometric distances from the closest classical state\cite{Bievre2019,Hillery1987, Dodonov2000, Marian2002}, the negativity of the Wigner function\cite{Kenfack2004}, or the amount of metrological advantage\cite{Kwon2019, Yadin2018}. However, these nonclassicality measures are frequently computationally intractable except in special cases, unable to detect every nonclassical state, or lack a physical interpretation. 

In this article, we propose a method to directly quantify the negativity of the $P$-function in a consistent way. It is based on the nonclassicality filtering approach proposed in Refs.\cite{Kiesel2010}. We show that this approach leads to a nonclassicality measure that will always decrease under linear optical operations, otherwise called a linear optical monotone. It is therefore a nonclassicality measure under the operational resource theory of nonclassicality proposed in Ref.\cite{Tan2017}. The measure also has a direct physical interpretation as the robustness of nonclassicality; it is the minimum amount of statistical mixing with classical noise that is needed to erase the nonclassicality of the state. We also demonstrate that the negativity of every $s$-parametrized quasiprobability\cite{Cahill1969} is not only a lowerbound to the negativity of the $P$-function, they are also themselves linear optical monotones. The set of $s$-parametrized quasiprobabilities therefore form a continuous hierarchy of nonclassicality measures. Finally, we propose an approximate nonclassicality monotone that is numerically computable for an arbitrary quantum state.

\section{Preliminaries}
We first introduce the characteristic function of the Glauber-Sudarshan $P$ function. A common convention is to define it as the integral $\int d^2\alpha P(\alpha) \exp[2i(\beta_i\alpha_r - \beta_r\alpha_i)] $, where $\alpha_r,\beta_{r}$ and $\alpha_i,\beta_i$ are the real and imaginary components of $\alpha$ and $\beta$ respectively. One may observe that this just a multivariate Fourier transformation. For our purposes, we will adopt the following convention: $$\chi(\beta) \coloneqq \int d^2\alpha P(\alpha) \exp[2\pi i(\beta_i\alpha_r + \beta_r\alpha_i)].$$ 

It should be clear that this definition essentially corresponds to a change in variables of the type $\beta_i \rightarrow \pi \beta_i'$ and $\beta_r \rightarrow -\pi \beta_r'$, and so does not alter the information content of the characteristic function. It also adheres more closely to the conventional definition of the Fourier transform in the ordinary frequency domain: $\mathcal{F}f(y) \coloneqq \int dx f(x)\exp(-2\pi i xy)$. The corresponding inverse Fourier transform is then $\mathcal{F}^{-1}f(y) \coloneqq \int dx f(x)\exp(2\pi i xy)$. This definition allows us to write $P(\alpha) = \mathcal{F}\chi(\alpha)$. All physical characteristic functions satisfies $\abs{\chi(\beta)} \leq \exp(\pi^2\abs{\beta}^2/2)$.

One major issue with the $P$-function is that it is frequently highly singular. This complicates our ability to analyze and quantify the nonclassicality of a quantum state via the $P$-function alone, and necessitates the use of other nonclassicality criteria. 

We consider the filtered $P$-functions proposed in Ref.~\cite{Kiesel2010}. Filtered $P$-functions are based on the observation that $P(\alpha)$ is the (multivariate) Fourier transform of the characteristic function $\chi(\alpha)$, such that $P(\alpha)= \mathcal{F}\chi(\alpha)$. This opens up the possibility of applying a filtering function $\Omega_w(\alpha)$ prior to the Fourier transform. The filtered function is then $$P_{\Omega,w}(\alpha) \coloneqq \mathcal{F}\chi_{\Omega,w}(\alpha)$$ where $\chi_{\Omega,w}(\beta) \coloneqq \chi(\beta)\Omega_w(\beta)$. In general, characteristic, $P$ and filtered $P$-functions depend on the state $\rho$. When the state $\rho$ is unambiguous, the characteristic function is denoted $\chi$ and $\chi(\alpha)$ is the function at the point $\alpha$. When $\rho$ needs to be specified, the characteristic function is denoted $\chi(\rho)$, while $\chi(\alpha \mid \rho)$ is the function at $\alpha$. Similar notations will also be used for the original and filtered $P$-functions.

The filter $\Omega_w$ must be carefully chosen. For our purpose, we require that they satisfy the following properties:

\begin{enumerate}[(a)]

\item $\Omega_w(\beta)$ is factorizable such that $\Omega_w(\beta) = \Omega^1_w(\beta) \Omega^2_w(\beta)$ s.t. $\Omega^i_w(\beta)$ is square integrable for $i=1,2$.
\item $\Omega^1_w(\beta)e^{\pi^2\abs{\beta}^2/2}$ is square integrable.
\item $\Omega_w(0)= 1$ and  $\lim_{w\rightarrow \infty}\Omega_w(\beta) =1$.
\item There exists $t>0$ s.t. $\Omega_{w}(\beta) = \Omega_{w/\abs{r}}(\beta)\Omega_{t}(\beta) $ for any $\abs{r}<1$, and some $t>0$.
\item $\Omega_{w}(\beta) = \Omega_{kw}(k\beta)$ for any $k>0$.

\end{enumerate}

Note that these conditions are stronger than those proposed in Ref.~\cite{Kiesel2010}. There, the key requirement is for $\Omega_w(\beta)e^{\pi^2\abs{\beta}^2/2}$ be square integrable, in order to ensure that its Fourier transform will also be square integrable due to Plancherel's theorem. Square integrability is however not sufficient to ensure that $P_{\Omega,w}(\alpha)$ is pointwise finite for every $\alpha$. Our modified approach closes this gap by ensuring that  $P_{\Omega,w}(\alpha)$ is always finite, which allows us to numerically determine whether there is negativity at a given point $\alpha$.

\begin{theorem} \label{thm::finite}
If $\Omega_w$ satisfies properties (a) and (b), then $P_{\Omega,w}(\alpha)$ contains no singularities and is finite for every $\alpha$.
\end{theorem}

\begin{proof}
Since $\chi_{\Omega,w}(\beta) \equiv \chi(\beta)\Omega_w(\beta) = \chi(\beta)\Omega^1_w(\beta) \Omega^2_w(\beta)$, we can group the terms such that $\chi_{\Omega,w}(\beta) = [\chi(\beta)\Omega^1_w(\beta)]\times \Omega^2_w(\beta)$. The convolution theorem then implies that $P_{\Omega,w}(\alpha) \equiv\mathcal{F}\chi_{\Omega,w}(\alpha) = \{\mathcal{F}[\chi(\beta)\Omega^1_w(\beta)]*\mathcal{F}\Omega^2_w(\beta)\}(\alpha)$. 

From property (a), we already know that $\Omega^2_w(\beta)$ and hence $\mathcal{F}\Omega^2_w(\beta)$ are square integrable from Plancherel's theorem. Furthermore, from property (b), we are guaranteed that $\Omega^1_w(\beta)e^{\pi^2\abs{\beta}^2}$ is square integrable. This means that $\Omega^1_w(\beta)\chi(\beta)$ is also square integrable since $\chi(\beta) \leq e^{\pi^2\abs{\beta}^2}$. Applying Plancherel's theorem again, we know that $\mathcal{F}[\chi(\beta)\Omega^1_w(\beta)]$ is also square integrable.

We recall that if $f(\beta)$ and $g(\beta)$ are both square integrable, then by Cauchy's inequality, it must satisfy $\lVert f(\beta) g(\beta) \rVert_1 \leq \lVert f(\beta) \rVert_2 \lVert g(\beta) \rVert_2$ where $\lVert \cdot \rVert_1$ and $\lVert \cdot \rVert_2$ are the $L_1$ and $L_2$ norms respectively.  Furthermore, since the $L_1$ norm is just the absolute integral, we have $\abs{\int d^2\beta f(\beta)g(\beta)} \leq \lVert f(\beta) \rVert_2 \lVert g(\beta) \rVert_2 < \infty$. This implies that the integral $\int d^2\beta f(\beta)g(\beta)$ is finite. 
 
$\{\mathcal{F}[\chi(\beta)\Omega^1_w(\beta)]*\mathcal{F}\Omega^2_w(\beta)\} (\alpha)$ is a convolution of two square integrable functions. By the definition of a convolution, for every given $\alpha$, it is an integral of a product of 2 square integrable functions. From the property described in the previous paragraph, we must have $P_{\Omega,w}(\alpha) \equiv \mathcal{F}\chi_{\Omega,w}(\alpha) = \{\mathcal{F}[\chi(\beta)\Omega^1_w(\beta)]*\mathcal{F}\Omega^2_w(\beta)\}(\alpha) < \infty$, so it has a finite value for every $\alpha$. This means that the filtered function $P_{\Omega,w}(\alpha)$ is finite everywhere and contains no singularities.
\end{proof}

Theorem~\ref{thm::finite} thus allows us to to assign definite positive or negative values to every point $\alpha$ of $P_{\Omega,w}(\alpha)$. This implies that we can determine unambiguously the positive and negative regions of $P_{\Omega,w}(\alpha)$. As such, for every $w$ the negative volume of $P_{\Omega,w}(\alpha)$ is well defined. Property (c) then guarantees that the filtered function is a proper quasiprobability function such that $\int d^2\alpha P_{\Omega,w}(\alpha) = 1$, and that for sufficiently large $w$, $\mathcal{F}\Omega_w(\alpha) \approx \delta(\alpha)$, so the original $P$-function is retrieved. This allows us to define the negativity of a $P$-function.

\begin{definition} [Negativity of a $P$ function] \label{def::NegP}

Let $f(\alpha)$ be a function that is well defined for every $\alpha$, so that we can write $f(\alpha)= f^+(\alpha)-f^-(\alpha)$, where $f^\pm(\alpha)$ are pointwise nonnegative functions. Then the negativity of $f$ is defined as $$\mathcal{N}(f) \coloneqq \int d^2\alpha f^-(\alpha).$$

Consider the $P$ function of a state $\rho$. Let $\Omega_w$ be some filter that satisfies properties (a)-(c). We can then write the filtered $P$-function as $P_{\Omega,w}(\alpha)= P^+_{\Omega,w}(\alpha)- P^-_{\Omega,w}(\alpha)$ where $P^\pm_{\Omega,w}(\alpha)$ are the nonnegative functions.

The negativity, of $\rho$ is defined to be $$\mathcal{N}(\rho) \coloneqq \lim_{w\rightarrow \infty} \int d^2\alpha P^-_{\Omega,w}(\alpha).$$  
\end{definition}

Given the above definition, we still need to find an appropriate filter $\Omega_w$. The astute reader may have noticed that properties (d) and (e) are not yet discussed. They will play an important role which will be described in greater detail in a subsequent section. We will first establish several properties of the negativity.

\section{Negativity as a linear optical monotone}

In Ref.~\cite{Tan2017}, a resource theoretical approach was proposed to quantify nonclassicality in radiation fields. There, it was argued that nonclassicality measures should be linear optical monotones, i.e. a nonclassicality should be measured using quantities that do not increase under linear optical maps. Given this approach, we can consider nonclassicality as potential resources to overcome the limitations of linear optics.

Linear optical maps are formally defined to be any quantum map that can be written in the form
$$\Phi_L(\rho_A) := \mathrm{Tr}_E [ U_L (\rho_A \otimes \sigma_E) U_L^\dag ],$$ 
where $\sigma_E$ is a classical state and $U_L$ is a linear optical unitary composed of any combination of beam splitters, phase shifters and displacement operations. Such unitary transforms will always map a $N$ mode bosonic creation operator $a_{\vec{\mu}}^\dag \coloneqq \sum_{i=1}^N \mu_i a_i^\dag$ into $a_{\vec{\mu'}}^\dag + \oplus_{i-1}^N \alpha_i\openone_ni$ where $\vec{\mu},\vec{\mu}'$ are $N$ dimensional complex vectors of unit length, and $\openone_i$ is the identity operator on the $i$th mode.

One may also incorporate postselection into the definition by defining selective linear optical operations via a set of Kraus operators $K_i$ for which there exists linear optical unitary $U_L$, classical ancilla $\sigma_{EE'}$, and a set of orthogonal vectors $\{\ket{i}_{E'} \}$ such that $\mathrm{Tr}_E [ U_L (\rho_A \otimes \sigma_{EE'}) U_L^\dag ]= \sum_i p_i \rho^i_{A} \otimes \ket{i}_{E'}\bra{i}$, where $p_i \rho^i_{A} \coloneqq {K}_i \rho_A K^\dag_i $ and $p_i \coloneqq \mathrm{Tr}(K_i \rho_A K^\dag_i)$. 

Based on this definition of linear optical maps, the following theorem shows that the negativity $\mathcal{N}$ is a linear optical monotone and therefore belongs to the operational resource theory outlined in Ref.~\cite{Tan2017}.

\begin{theorem}
\label{thm::monotonicity}
The negativity ${\cal N}(\rho)$ is a faithful nonclassicality measure satisfying the following properties:
\begin{enumerate}
\item $\mathcal{N}(\rho) = 0$ iff  $\rho$ has a classical $P$-function.

\item
	\begin{enumerate}
	
	\item  (Weak monotonicity) $\mathcal{N}(\rho)\geq \mathcal{N}(\Phi_L (\rho))$.
	
	\item (Strong monotonicity) $\mathcal{N}(\rho) \geq \sum_i p_i\mathcal{N}(\rho_i)$ where $p_i \coloneqq \mathrm{Tr}(K^\dag_iK_i \rho )$ and $\rho_i \coloneqq (K_i \rho K^\dag_i)/ p_i$.
	
	\end{enumerate}

\item (Convexity), i.e. $\mathcal{N}(\sum_i p_i \rho_i) \leq \sum_i p_i \mathcal{N}( \rho_i)$.

\end{enumerate}

\end{theorem}

\begin{proof}
It is apparent that if the $P$-function of $\rho$ is classical, then $\mathcal{N}(\rho)$ = 0 since $P_{\Omega,w}(\alpha) \rightarrow P(\alpha)$ as $w \rightarrow \infty$ so the negative volume must vanish. The converse must also be true as if $\mathcal{N}(\rho) = 0$, then $\int d^2\alpha P^-_{\Omega,w}(\alpha) \rightarrow 0$ as $w \rightarrow \infty$, which implies $P^+_{\Omega,w}(\alpha) \rightarrow P(\alpha)$. This means that  $P(\alpha)$ is the limit of a sequence of positive distributions. As the set of classical states is a closed convex set, and $P^+_{\Omega,w}(\alpha) \rightarrow P(\alpha)$, this means that $P(\alpha)$ must be classical. This proves Property 1.

In order to prove the weak and strong monotonicity properties, we make use of an observation from Ref~\cite{Tan2017}. It was noted that for the special case when the $P$ is any regular function that does not contain any singularities, the negative volume $\mathcal{N}(\rho)$ where $\rho =  \int d^2\alpha P(\alpha) |\alpha\rangle \langle \alpha|$ satisfies both weak and strong monotonicity conditions.

We now extend the above result to all $P$ functions. Let $\rho_{w} = \int d^2\alpha P_{\Omega,w}(\alpha) |\alpha\rangle\langle\alpha|$. For weak monotonicity, we see that $\mathcal{N}(\rho_w)\geq \mathcal{N}(\Phi_L (\rho_w))$. Taking the limit $w \rightarrow \infty$, $\rho_w \rightarrow \rho$, the inequality converges to $\mathcal{N}(\rho)\geq \mathcal{N}(\Phi_L (\rho))$. Identical arguments hold for strong monotonicity. This is sufficient to generalize the monotonicity property to all $P$ functions, and establishes Property 2. 

Similarly for convexity, we have $\mathcal{N}(\sum_i p_i \rho_{w,i}) \leq \sum_i p_i \mathcal{N}( \rho_{w,i})$ for every $w$. Taking the limit $w \rightarrow \infty$, $\rho_{w,i} \rightarrow \rho_i$ so the inequality converges to $\mathcal{N}(\sum_i p_i \rho_i) \leq \sum_i p_i \mathcal{N}( \rho_i)$ which is the required inequality.

%Finally, in order to establish convexity, let us consider two filtered distributions $P_{\Omega,w}(\alpha)$ and $Q_{\Omega,w}(\alpha)$ and their convex mixture $R_{\Omega,w}(\alpha) = pP_{\Omega,w}(\alpha)+(1-p)Q_{\Omega,w}(\alpha)$ where $p \in [0,1]$. It is apparent that the negative portion $R_{\Omega,w}(\alpha)$ can only come from the negative portions of $P_{\Omega,w}(\alpha)$ and $Q_{\Omega,w}(\alpha)$, so integrating over the negative volume, one is able to get the crude upper bound $\int d^2  R^-_{\Omega,w}(\alpha) \leq p \int d^2 P^-_{\Omega,w}(\alpha)+(1-p) \int d^2 Q^-_{\Omega,w}(\alpha)$. By taking the limit $w\rightarrow \infty $, we get $\mathcal{N}(p \rho + (1-p)\sigma) \leq p \mathcal{N}(\rho)+(1-p) \mathcal{N}(\sigma)$ where $P_{\Omega,w}(\alpha)\rightarrow P$ and $Q_{\Omega,w}(\alpha) \rightarrow Q$ are distributions corresponding to the states $\rho$ and $\sigma$ respectively. So $\mathcal{N}$ is a convex functional of state, which proves Property 3.

\end{proof}

\section{Equivalence between negativity and robustness}

An operational measure that has been extensively studied in various quantum resource theories is the robustness\cite{Vidal1999, Napoli2016}. It quantifies the minimum amount of mixing with noise that is necessary to make a given quantum state classical. It turns out that the negativity exactly quantifies the robustness of a given quantum state.

We can consider the following definition for the robustness of nonclassicality.

\begin{definition}[Robustness of nonclassicality] 
Let $\mathcal{P}$ be the set of all quantum states with classical $P$ distributions.

The robustness of nonclassicality is defined as $$\mathcal{R}(\rho) \coloneqq \min_{\sigma \in \mathcal{P}}\{r \mid r\geq 0, \frac{\rho+r \sigma}{1+r} \in \mathcal{P} \}.$$

\end{definition} 

Based on the above definition, one may show that the negativity and the robustness are in fact equivalent.

\begin{theorem} \label{thm::robustness}

The negativity and the robustness are equivalent measures of nonclassicality, i.e. $\mathcal{N}(\rho) = \mathcal{R}(\rho)$ for every quantum state $\rho$.

\end{theorem}

\begin{proof}

First, note that we can always write $P_{\Omega,w}(\alpha)= P^+_{\Omega,w}(\alpha)- P^-_{\Omega,w}(\alpha)$ where $P^\pm_{\Omega,w}(\alpha)$ are pointwise nonnegative functions. Let $\int d^2 \alpha P^-_{\Omega,w}(\alpha) \coloneqq r_w$ and $\lim_{w\rightarrow \infty}r_w = r$. Note that by this definition, $r = \mathcal{N}(\rho)$ 

We now consider some sufficiently large $w$ and observe that $r$ is always an upper bound to the robustness. This is because, $\frac{1}{1+r_w}(P_{\Omega,w}(\alpha)+P^-_{\Omega,w}(\alpha))=\frac{1}{1+r_w}P^+_{\Omega,w}(\alpha)$ which corresponds to a positive, and hence classical, $P$-function. Therefore, if $\rho_w$ and $\sigma_w$ are the quantum states corresponding to the distributions $P_{\Omega,w}(\alpha)$ and $P^-_{\Omega,w}(\alpha)/r_w$ respectively, the mixture $\frac{\rho_w+r_w\sigma_w}{1+r_w}$ always has a classical $P$-function. Taking the limit $w \rightarrow \infty$, we get $\frac{\rho+r\sigma}{1+r}$ is classical, where $\sigma \coloneqq \lim_{w\rightarrow \infty}\sigma_w$. Since $r$ is just the negativity $\mathcal{N}(\rho)$, we see that the negativity is at least an upper bound to the robustness.

We now need to show that $r$ is also a lower bound. This follows immediately from the observation that $P^-_{\Omega,w}(\alpha)$ is the minimal function necessary for $P_{\Omega,w}(\alpha)$ to be positive. It is clear that if $P'(\alpha) < P^-_{\Omega,w}(\alpha)$ for any $\alpha$, then $P_{\Omega,w}(\alpha)+ P'(\alpha) < 0 $ and so is not positive at $\alpha$. This shows that $r$ must also be a lower bound and proves the theorem.

\end{proof}

\section{Relationship with the negativity of other quasiprobabilities}

It is well known that the characteristic function of $P$ is related to the characteristic functions of other commonly studied quasiprobability distributions via the following relation: $$\chi_s(\beta) \coloneqq \chi(\beta)e^{-(1-s)\pi^2\abs{\beta}^2/2}.$$

Note that this differs slightly from the usual convention due to the convention we employ for $\chi(\beta)$. For $s =1$, we retrieve the characteristic function of the $P$-function, for $s = 0$ the characteristic function leads to the Wigner function, while for $s=-1$ the characteristic function is related to the Husimi $Q$ function. These form the set of $s$-parametrized quasiprobability distributions\cite{Cahill1969}. 

We can define the negative volume of the $s$-parametrized quasiprobabilities using a similar approach.

\begin{definition}[$s$-parametrized negativity] \label{def::NegPs}
Let $P_s(\alpha) \coloneqq \mathcal{F}\chi_s(\alpha)$ be some $s$-parametrized quasiprobability, and let $P_{s,\Omega,w}(\alpha) \coloneqq \mathcal{F}\chi_{s,\Omega, w}(\alpha)$ be the filtered $s$-parametrized quasiprobability, where $\chi_{s,\Omega, w}(\beta) \coloneqq  \chi_s(\beta)\Omega_w(\beta)$ for some filter $\Omega_w$ satisfying properties (a)-(c).

We can then write $P_{s,\Omega,w}(\alpha)= P^+_{s,\Omega,w}(\alpha)- P^-_{s,\Omega,w}(\alpha)$ where $P^\pm_{s,\Omega,w}(\alpha)$ are well defined.

The $s$-parametrized negativity is defined as $$\mathcal{N}_s(\rho) \coloneqq \lim_{w \rightarrow \infty}\int d^2 \alpha P^-_{s,\Omega,w}(\alpha).$$
\end{definition}

Given the above definition, we can establish several properties. The following theorem establishes the monotonic dependence of $\mathcal{N}_s$ on $s$.

\begin{theorem} \label{thm::sMonotonicity}
$\mathcal{N}_s(\rho)$ is a monotonically increasing function of $s\leq 1$ and is upper bounded by the negativity of the $P$-function, i.e. $\mathcal{N}_s(\rho) \leq \mathcal{N}(\rho)$.
\end{theorem}

\begin{proof}

First, we show that $\mathcal{N}_s(\rho)$ is a monotonically decreasing function of $s$ for any given $\rho$. 

First, note that $\mathcal{F}(e^{-a\abs{\beta}^2})(\alpha)= \frac{\pi}{a}e^{-(\pi \abs{\alpha})^2/a} = N_\frac{a}{2\pi^2}(\alpha)$ where $N_{\sigma^2}(\alpha)$ is the normalized Gaussian function with variance $\sigma^2$.

Second, we observe that the convolution of 2 normalized Gaussian functions just sums up the variance, i.e. $N_{a^2}*N_{b^2}(\alpha)=N_{a^2+b^2}(\alpha)$

Third, we observe that a convolution with a positive probability distribution function (PDF) can never increase the negativity. To see this, let $f(x)=f^+(x)-f^-(x)$  where $f^\pm(x)$ are non-negative functions that are well defined. Let $g(x)$ be a positive PDF. Then $f*g(x) = f^+*g(x)-f^-*g(x)$. It is then apparent that $\int dx(f*g)^-(x) \leq \int dx f^-*g(x)$ since $g(x)$ is pointwise positive. Finally, since $g$ is a PDF, $\int dx g(x) = 1$, we have $\int dx f^-*g(x) = \int dx f^-(x) \int dx g(x) =  \int dx f^-(x)$ which is just the negativity of $f(x)$.  This shows that the negativity never increases under convolution with a PDF.

Let $q=(1-s)\pi^2/2 \geq 0$, so we can consider instead $\chi'_q(\beta) \coloneqq \chi(\beta)e^{-q\abs{\beta}^2}$, $G'_q(\alpha) \coloneqq \mathcal{F}\chi'_q(\alpha)$ and $\mathcal{N}'_q(\rho) = \int d^2 G'^-_q(\alpha)$. By the convolution theorem, we know that $\mathcal{F}\chi'_q(\alpha) = \mathcal{F}\chi*N_{\frac{q}{2\pi^2}}(\alpha)$. Furthermore, for any $q_1,q_2$ satisfying $q_1+q_2=q$, we can always consider the decomposition $G'_q(\alpha)=\mathcal{F}\chi'_q(\alpha) = \mathcal{F}\chi*N_{\frac{q_1}{2\pi^2}}*N_{\frac{q_2}{2\pi^2}}(\alpha)= G'_{q_1}*N_{\frac{q_2}{2\pi^2}}(\alpha)$. Since $N_{\sigma^2}(\alpha)$ is a properly normalized PDF, and we know that a convolution with a PDF cannot increase negativity, this shows that $\mathcal{N}'_{q_1}(\rho)\geq \mathcal{N}'_{q}$ when $q_1 \leq q$. Finally, since $q$ monotonically decreases with $s$, this means $\mathcal{N}_s(\rho)$ monotonically increases with $s$. This proves the first part of the theorem.

Finally, to see that $\mathcal{N}_s(\rho) \leq \mathcal{N}(\rho)$, we just observe that at $s = 1$, we retrieve $\mathcal{N}_{s=1}=\mathcal{N}$. From the monotonicity property above, we then have $\mathcal{N}_s(\rho) \leq \mathcal{N}(\rho)$ for $s \leq 1$.

%Finally, we just need to show that $\mathcal{N}_s(\rho) \leq \mathcal{N}(\rho)$, we just need to consider the state $\rho_w$ which has $P_{\Omega,w}(\alpha)$ as its $P$-function over sufficiently large $w$. For every $\rho_w$, it satisfies $\mathcal{N}_s(\rho_w) \leq \mathcal{N}(\rho_w)$, so in the limit $w \rightarrow \infty$ and $\rho_w \rightarrow \rho$, we get $\mathcal{N}_s(\rho) \leq \mathcal{N}(\rho)$ which is the required inequality.

\end{proof}

We can interpret the $s$-parametrized quasiprobability distributions as the $P$-function with a Gaussian filter applied. In general, as $s$ decreases, the width of the applied Gaussian filter increases, which also decreases any observed negativity. Ultimately, any negativity that is observed in any $s$-parametrized quasiprobability function originates from the negativity of the Glauber-Sudarshan $P$-function itself. 

It is therefore natural to ask if the negativity of the $s$-parametrized quasiprobabilities is a nonclassicality measure that monotonically decreases under linear optical operations. The following theorem affirms this fact.

\begin{theorem}\label{thm::sNegMono}

The $s$-parametrized negativity $\mathcal{N}_{s}(\rho)$ is a nonclassicality measure satisfying the following properties:
\begin{enumerate}
\item $\mathcal{N}_{s}(\rho) = 0$ if  $\rho$ has a classical $P$-function.

\item
	\begin{enumerate}
	
	\item  (Weak monotonicity) $\mathcal{N}_{s}(\rho)\geq \mathcal{N}_{s}(\Phi_L (\rho))$.
	
	\item (Strong monotonicity) $\mathcal{N}_{s}(\rho) \geq \sum_i p_i\mathcal{N}_{s}(\rho_i)$ where $p_i \coloneqq \mathrm{Tr}(K^\dag_iK_i \rho )$ , $\rho_i \coloneqq (K_i \rho K^\dag_i)/ p_i$ and $\Phi_L (\rho) = \sum_i K_i \rho K^\dag_i$ is a selective linear optical operation.
	
	\end{enumerate}

\item (Convexity), i.e. $\mathcal{N}_{s}(\sum_i p_i \rho_i) \leq \sum_i p_i \mathcal{N}_{s}( \rho_i)$ .

\end{enumerate}

\end{theorem}

\begin{proof}
Property 1 immediately follows from the fact that $N_s$ obtained from the Gaussian convolution of the $P$-function. Since the $P$-function of a classical state is pointwise positive, a convolution with a Gaussian function, which is itself also pointwise positive, cannot produce a negativity.

Property 3 follows from the convexity of $\mathcal{N}$. We observe that the $s$ paramatrized quasiprobabilities of any given state are themselves physical $P$-functions, so $\mathcal{N}_s$ must be convex if $\mathcal{N}$ is convex.

Proving the weak and strong monotonicity properties will first require us to gather several facts. Let $\Phi_L$ be some linear optical unitary. By definition, this means we can write $$\Phi_L(\rho_A) := \mathrm{Tr}_E [ U_L (\rho_A \otimes \sigma_E) U_L^\dag ].$$ Also recall that $U_L$ is a linear optical unitary, and so will map a $N$ mode bosonic creation operator $a_{\vec{\mu}}^\dag \coloneqq \sum_{i=1}^N \mu_i a_i^\dag$ into the form $a_{\vec{\mu'}}^\dag + \oplus_{i-1}^N \alpha_i\openone_i$. $\vec{\mu} \rightarrow \vec{\mu'}$ represents a rotation in $N$ dimensional complex space (otherwise called an $SU(N)$ interferometer~\cite{Reck1994}), while $\alpha_i$ represents linear displacements in phase space on the $i$th mode.  We will assume the index $i=1$ denotes the mode of the system of interest $A$, with the other indices representing the rest of the ancillary modes. We will also denote the superoperator of the linear optical unitary $U_L$ as $\mathcal{U}_L = U_L (\cdot) U_L^\dag$.

Consider the displacement operator $D(\beta)$ acting on mode 1. This performs the map $a^\dag_1 \rightarrow a^\dag_1+\beta\openone_1$. If we have ancillary modes, it is a linear displacement in the direction $\vec{\beta} = (\beta, 0, \ldots, 0)$ in N dimensional complex parameter space. For any direction $\vec{\alpha} = (\alpha_1, \ldots, \alpha_N)$, let $\mathcal{D}_{\vec{\alpha}}(\cdot) \coloneqq D_1(\alpha_1)\ldots D_N(\alpha_N) (\cdot) D_1^\dag(\alpha_1)\ldots D_N^\dag(\alpha_N)$. 

In complex parameter space, $\mathcal{D}_{\vec{\beta}} \circ \mathcal{U}_L$ corresponds to a displacement, followed by a unitary rotation, followed by another displacement, i.e. $\Delta_{\vec{\beta}}  U \Delta_{\vec{\alpha}}$, where $\Delta_{\vec{\beta}}$ is a displacement in direction $\vec{\beta}$, and $U$ is a (unitary) rotation. Displacements commute, so $\Delta_{\vec{\beta}} \Delta_{\vec{\alpha}} = \Delta_{\vec{\alpha}} \Delta_{\vec{\beta}}$. Furthermore, a unitary rotation followed by displacement is the same as a rotated displacement followed by a unitary rotation, i.e. $\Delta_{\vec{\beta}} U = U \Delta_{U^\dag\vec{\beta}}$. As a result, we have $\Delta_{\vec{\beta}} ( U  \Delta_{\vec{\alpha}})= (U \Delta_{\vec{\alpha}} ) \Delta_{U^\dag\vec{\beta}}$. This implies that $\mathcal{D}_{\vec{\beta}} \circ \mathcal{U}_L = \mathcal{U}_L \circ \mathcal{D}_{U'(\Phi_L)\vec{\beta}}$, where $U'(\Phi_L)$ is some unitary depending on $\Phi_L$.

We make use of two other observations. First, in Ref.~\cite{Tan2017} it was noted that when the $P$-function of $\rho$ is a regular function that does not contain any singularities, the negative volume $\mathcal{N}(\rho)$ satisfies both weak and strong monotonicity conditions.

Second, in Ref.~\cite{Kuhn2018} it was observed that the filtered function $P_{\Omega,w}$ is the output of an interaction with an ancilla and a highly transmissive beam splitter. When the Fourier transform of the filter $\Omega_w$ is pointwise positive, the filtering operation is actually a linear optical map $\Phi_{\Omega,w}$, which maps an initial $P$-function to the filtered $P$-function $P_{\Omega,w}$. Furthermore, the filtering operation can be interpreted as a stochastic displacement operation $\Phi_{\Omega,w}=\sum_i p_i \mathcal{D}_{\vec{\alpha}_i}$ with probability distribution $p_i$ sampled from the probability density function $\mathcal{F}\Omega_w(\alpha)$. If we choose the filter to be the Gaussian filter $\Omega_w(\beta) = e^{-\abs{\beta/w}^2} =  e^{-(1-s)\pi^2\abs{\beta}^2/2} $ where $1/w =(1-s)\pi^2/2 $, then we see that $\mathcal{N}(\Phi_{\Omega,w}(\rho)) = \mathcal{N}_s(\rho)$. 

Choosing $\mathcal{F}\Omega_w$ to be a normalized, Guassian PDF, we can obtain the following series of inequalities:

\begin{align}
 \mathcal{N}[\Phi_{\Omega,w}(\rho)] &\geq  \mathcal{N}[\Phi_L\Phi_{\Omega,w}(\rho)] \label{eqn::s1}\\
 &= \mathcal{N}[\Phi_L\sum_i p_i\mathcal{D}_{\vec{\alpha}_i}(\rho)] \label{eqn::s2}\\
 &= \mathcal{N}\{\Tr_a[V_L\sum_i p_i\mathcal{D}_{\vec{\alpha}_i}(\rho \otimes \sigma_a) V_L^\dag]\} \label{eqn::s3}\\
 &= \mathcal{N}\{\Tr_a[\sum_i p_i\mathcal{D}_{U(\Phi_L)\vec{\alpha}_i}(V_L\rho\otimes \sigma_a V_L^\dag)]\} \label{eqn::s4}\\
  &= \mathcal{N}[\sum_i p_i\mathcal{D}_{\abs{r}\vec{\alpha}_i}\Phi_L(\rho) ]  \label{eqn::s5}\\
  &= \mathcal{N}[\Phi_{\Omega,w/\abs{r}}\Phi_L(\rho)]  \label{eqn::s6}\\
  &\geq \mathcal{N}[\Phi_{\Omega,w}\Phi_L(\rho)] \label{eqn::s7}
\end{align}

Eqn~\ref{eqn::s1} comes from the fact that the negativity of the $P$-function, $\mathcal{N}$, is a linear optical monotone, and that both $\Phi_L$ and $\Phi_{\Omega,w}$ is a linear optical map (see Ref.~\cite{Kuhn2018}). Eqn~\ref{eqn::s2} uses the decomposition of $\Phi_{\Omega,w}$ into a stochastic displacement operation $\Phi_{\Omega,w}=\sum_i p_i \mathcal{D}_{\vec{\alpha}_i}$. Eqn~\ref{eqn::s3} follows from the definition of a linear optical map. Eqn~\ref{eqn::s4} uses the relation $\mathcal{D}_{\vec{\beta}} \circ \mathcal{U}_L = \mathcal{U}_L \circ \mathcal{D}_{U'(\Phi_L)\vec{\beta}}$ since $V_L$ is a linear optical unitary. Eqn~\ref{eqn::s5} comes from the fact that $\vec{\alpha_i} = (\alpha_i, 0 ,\ldots, 0)$ is a displacement on the first mode $U(\Phi_L)\vec{\alpha_i} = \alpha_i(u_1, \ldots, u_n)$ where $\sum_{i=1}^n \abs{u_i}^2=1$. Setting $u_1 = r$ and observing that a phase rotation does not change the negativity leads to the required equality. Eqn~\ref{eqn::s6} comes from the observation that $\abs{r} \leq 1$. Assuming $\Phi_{\Omega,w}=\sum_i p_i \mathcal{D}_{\vec{\alpha}_i}$ is a Gaussian filter where $p_i$ is sampled from a normalized Gaussian distribution, $\sum_i p_i \mathcal{D}_{\abs{r}\vec{\alpha}_i}$ is a Gaussian filter with scaling factor $\frac{1}{\abs{r}}$, i.e.   $\Phi_{\Omega,w/\abs{r}} = \sum_i p_i \mathcal{D}_{\abs{r}\vec{\alpha}_i}$. The inequality in Eqn~\ref{eqn::s7} follows from the fact that the $s$-parametrized negativity monotonically increases with $s$ and hence $w$ (see Theorem~\ref{thm::sMonotonicity}). The final inequality then gives us $\mathcal{N}_s(\rho) \geq \mathcal{N}_s[\Phi_L(\rho)]$. This proves the weak monotonicity property.

The strong monotonicity property follows from largely the same arguments up until Eqn~\ref{eqn::s6}. From there, we have the following series of inequalities:

\begin{align}
 \mathcal{N}[\Phi_{\Omega,w}(\rho)] &\geq  
   \mathcal{N}[\Phi_{\Omega, w/\abs{r}}\Phi_L(\rho)]  \\
   &= \mathcal{N}[\sum_{i} \Phi_{\Omega,w/\abs{r}}(K_i \rho K_i^\dag)]  \label{eqn::s8}\\
   &\geq \sum_{i}\mathcal{N}[ \Phi_{\Omega,w/\abs{r}}(q_i\rho_i)] \label{eqn::s9}\\ 
   &\geq  \sum_{i}q_i\mathcal{N}[ \Phi_{\Omega,w}(\rho_i)] \label{eqn::s10}
\end{align}

In Eqn~\ref{eqn::s8}, we write the selective linear map $\Phi_L$ in terms of its Kraus decompositions $\Phi_L(\rho) = \sum_i L_i\rho L_i^\dag$. In Eqn~\ref{eqn::s9} we use the property that $\mathcal{N}$ is strongly monotonic, and denote $q_i = \Tr(K_i \rho K_i^\dag)$. In Eqn~\ref{eqn::s10} we used Theorem~\ref{thm::sMonotonicity} together with the fact that $w$ monotonically increases with $s$. The final inequality then gives us $\mathcal{N}_s(\rho) \geq \sum_i q_i \mathcal{N}_s(\rho_i)$, which proves strong monotonicity.

\end{proof}

Theorem~\ref{thm::sNegMono} therefore establishes that the set of $s$-parametrized negativities $\mathcal{N}_s$ forms a continuous hierarchy of nonclassicality measures under the operational resource theory of Ref.\cite{Tan2017}.

\section{Approximate nonclassicality monotones}

The negativity of quasiprobabilities are well defined in Definitions~\ref{def::NegP} and \ref{def::NegPs}. However they do not always lead to finite quantities. For instance, highly singular states such as squeezed states can possess infinite negativities. This can be verified numerically by applying an appropriate filter and computing the filtered negativities as $w \rightarrow \infty$.  From Theorem~\ref{thm::robustness}, we know that this is because some states require an infinite amount of statistical mixing with classical states before their nonclassicality is erased. Nevertheless, $\mathcal{N}_s$ remains a linear optical monotone. For $s=1$, we retrieve the negativity $\mathcal{N}$ of the $P$-function (see Definition~\ref{def::NegP}), which is a faithful nonclassicality measure. This means that the measure is able to unambiguously identify every nonclassical state. In contrast, for $s<1$, $\mathcal{N}_s$ corresponds to weaker nonclassicality measures as it may not be able to identify some nonclassical states. For instance, at $s=0$, $\mathcal{N}_s$ is the negativity of the Wigner function\cite{Kenfack2004}. It is a well known property of the Wigner function that its negativity cannot detect squeezed states. 

It is therefore natural to ask whether it is possible to avoid the aforementioned issues with infinite values while simultaneously maximizing the number of identifiable nonclassical states. In this section, we show that this is possible via an appropriate choice of filters that satisfies the full suite of properties (a)-(e) (see Preliminaries).

We begin with 2 lemmas that are particular consequences of properties (d) and (e).

\begin{lemma}
Suppose the filter $\Omega_w$ satisfies properties (a)-(d). Then for any $\abs{r}< 1$, $$\mathcal{N}(P_{\Omega, w}) \leq \mathcal{N}(P_{\Omega, w/\abs{r}}) [1+2 \mathcal{N}(\mathcal{F}\Omega_{t})]+\mathcal{N}(\mathcal{F}\Omega_{t}).$$
\end{lemma}

\begin{proof}

Let $f(\alpha) = f^+(\alpha)- f^-(\alpha)$, and $g(\alpha) = g^+(\alpha)- g^-(\alpha)$, where $f^\pm(\alpha)$ and $g^\pm(\alpha)$ are pointwise nonnegative functions. We also assume that for $f, g$ are normalized such that $\int d^2\alpha f(\alpha) = \int d^2\alpha g(\alpha) = 1$.

We note that $f*g = (f^+ - f^-)*(g^+ - g^-) = f^+*g^+ + f^-*g^- -(f^+*g^- + f^-*g^+)$. As a result, we have the following series of inequalities

\begin{align*}
&\int d^2  \alpha(f*g)^-(\alpha) \\
&\leq \int d^2\alpha (f^+*g^-(\alpha) + f^-*g^+(\alpha)) \\
&= \int d^2\alpha f^+(\alpha) \int d^2\alpha g^-(\alpha) + \int d^2\alpha f^-(\alpha) \int d^2\alpha g^+(\alpha) \\
&= (1+ \int d^2\alpha f^-(\alpha)) \int d^2\alpha g^-(\alpha) \\ & \qquad \qquad + \int d^2\alpha f^-(\alpha) (1+ \int d^2\alpha g^-(\alpha)) \\
&= \int d^2\alpha g^-(\alpha) (1 + 2 \int d^2\alpha f^-(\alpha)) + \int d^2\alpha f^-(\alpha),
\end{align*} where we used the identity $\int d^2\alpha f^+(\alpha)=1+ \int d^2\alpha f^-(\alpha)$ which comes from the fact that $\int d^2\alpha f(\alpha)=1$.

Since $\Omega_w$ satisfies properties (a)-(d), for any characteristic function $\chi$, we have $\mathcal{F}(\Omega_{w} \chi) = \mathcal{F}(\Omega_{w/\abs{r}}\Omega_{t}\chi) = \mathcal{F}(\Omega_{w/\abs{r}}\chi)*\mathcal{F}(\Omega_{t}) = P_{\Omega, w/\abs{r}}*\mathcal{F}(\Omega_{t})$.

We get the required expression by setting $f = \mathcal{F}(\Omega_{t})$ and $g = P_{\Omega, w/\abs{r}}$.

\end{proof}

\begin{lemma}
Suppose the filter $\Omega_w$ satisfies property (a)-(c) and (e). Then for any given linear optical map $\Phi_L$, $\mathcal N [P_{\Omega,w}(\rho)] \geq \mathcal N \{P_{\Omega,w}[\Phi_L(\rho)]\}$ where the factor $r$ depends only on $\Phi_L$.

\end{lemma}

\begin{proof}

Similar to an observation from Ref.~\cite{Kuhn2018} and the proof of Theorem~\ref{thm::sNegMono}, we define the map
\begin{equation*}
\Phi_{\Omega, w}(\rho) = \int d^2 \gamma \mathcal{F}\Omega_{w}(\gamma)D (\gamma)\rho D^\dagger (\gamma).
\end{equation*}

From this, we see that $P(\Phi_{\Omega, w}(\rho)) = P_{\Omega,w}(\rho)$, so the $P$-function after this map is equivalent to applying a filter $\Omega_w$. This property does not require $\mathcal{F}\Omega_w(\alpha)$ to be pointwise positive for every $\alpha$.

Using the notation $\mathcal{D}_\alpha (\cdot) = D (\alpha)(\cdot) D^\dagger (\alpha) $, we follow a similar argument with the proof of Theorem~\ref{thm::sNegMono}, resulting in the following series of inequalities:
\begin{align}
    \mathcal{N}[\Phi_{\Omega, w}(\rho)] &\geq \mathcal{N}[\Phi_L \Phi_{\Omega, w}(\rho)] \\
    &= \mathcal{N}[\Phi_L \int d^2 \alpha \mathcal{F}\Omega_{w} (\alpha) \mathcal{D}_\alpha (\rho)] \\
    &= \mathcal{N}[\int d^2 \alpha \mathcal{F}\Omega_{w} (\alpha)\mathcal{D}_{\abs{r}\alpha} \Phi_L (\rho)] \\
    &= \mathcal N [\int \frac{d^2 \alpha}{\abs{r}^2} \mathcal{F}\Omega_{w} (\frac{\alpha}{\abs{r}})\mathcal{D}_\alpha\Phi_L(\rho)] \\
   &= \mathcal N (\int d^2 \alpha \mathcal{F}\Omega_{w/\abs{r}} (\alpha) \mathcal D _\alpha \Phi_L (\rho))  \label{w-alpha} \\
    &= \mathcal N (\Phi_{\Omega, w/\abs{r}} \Phi_L (\rho)),
\end{align} where $\abs{r} \leq 1$ and depends only on $\Phi_L$. Which is the required expression.
Eqn. \ref{w-alpha} comes from the observation that whenever the filter satisfies $\Omega_{w}(\beta) = \Omega_{kw}(k\beta)$ for any $k>0$ (property (e)), then together with the scaling property $\mathcal{F}[f(\abs{r}\beta)](\alpha)=\mathcal{F}f(\alpha/\abs{r})/\abs{r}^2$ we have $$\mathcal{F}\Omega_w (\alpha /\abs{r}) = \abs{r}^2 \mathcal{F}\Omega _{w/\abs{r}}(\alpha).$$

\end{proof}

The above lemmas then imply the following bound for a finite $w$.

\begin{theorem}\label{thm::filtNegMono}
If the filter $\Omega_w$ satisfies properties (a)-(e), then for any given linear optical map $\Phi_L$, we have $$(1+2 \delta)\mathcal N [P_{\Omega, w}(\rho)]+\delta \geq \mathcal N \{P_{\Omega, w}[\Phi_L(\rho)] \}$$ where $\delta = \mathcal N (\mathcal{F}\Omega_{w=1})$.

\end{theorem}

\begin{proof}

Consider any given linear optical map $\Phi_L$. By Lemma 1, for a state $\Phi_L(\rho)$, we obtain \begin{align*}
&\mathcal{N}\{ P_{\Omega, w}[\Phi_L(\rho)]\} \\ &\leq \mathcal{N}\{P_{\Omega, w/\abs{r}}[\Phi_L(\rho)]\} [1+2 \mathcal{N}(\mathcal{F}\Omega_{t})]+\mathcal{N}(\mathcal{F}\Omega_{t})
\end{align*}
Combining the above and Lemma 2, we get the following inequalities:
\begin{align*}
\mathcal{N}[P_{\Omega, w}(\rho)] &\geq \mathcal{N}\{P_{\Omega, w/\abs{r}}[\Phi_L(\rho)]\} \\ &\geq \frac{\mathcal{N}\{P_{\Omega, w}[\Phi_L(\rho)]\}-\mathcal{N}(\mathcal{F}\Omega_t)}{1+2\mathcal \mathcal{N}(\mathcal{F}\Omega_t)}
\end{align*}
Define $\delta = \mathcal \mathcal{N}(\mathcal{F}\Omega_t)$. 

Finally, we observe that  because $\Omega_{w}(\alpha) = \Omega_{kw}(k\alpha)$ for any $k>0$, if we set $w=1$ and $k=t$, we get $\Omega_{w=1}(\alpha) = \Omega_{t}(t\alpha)$ . From the scaling property of the Fourier transform $\mathcal{F}[\Omega_{t}(t\beta)](\alpha)=\mathcal{F}\Omega_{t}(\alpha/t)/t^2$, we also have that  $\int d^2\alpha \mathcal{F}\Omega_{t}^-(\alpha/t)/t^2 = \int d^2\alpha \mathcal{F}\Omega_{t}^-(\alpha)t^2/t^2 = \mathcal{N}(\mathcal{F}\Omega_t)$. This implies that $\mathcal{N}(\mathcal{F}\Omega_{w=1}) = \mathcal{N}(\mathcal{F}\Omega_t)$ for any $t>0$, which completes the proof.
\end{proof}

Theorem~\ref{thm::filtNegMono} suggests that given a filter that satisfies properties (a)-(e), when the negativity of the Fourier transform of the filter is small, the filtered negativity $\mathcal N_{\Omega, w} (\rho)$ is approximately a linear optical monotone. Ideally, we would like the Fourier transform of the filter to be pointwise positive and still satisfy properties (a)-(e), which would imply that the filtered negativity is an exact linear optical monotone which can be computed for every $w>0$. It remains unclear whether this is possible, but we demonstrate that the negativity of the filter can at least be made arbitrarily small, such that the filtered negativity is essentially a linear optical monotone to any arbitrary level of precision.

\begin{proposition}\label{prop::filt}
Define $\Omega_{w, \epsilon}(\beta) \coloneqq \exp(-\abs{\beta/w}^{2+\epsilon})$, where $w>0$ is the width parameter, and $\epsilon>0$ is the error parameter.

Then $\Omega_{w, \epsilon}$ is a filter that satisfies properties (a)-(e). Furthermore, $\mathcal{N}(\mathcal{F}\Omega_{w=1, \epsilon}) \rightarrow 0 $ as $\epsilon \rightarrow 0$.
\end{proposition}

\begin{proof}

For property (a), we simply choose $\Omega_w^1(\beta)= \exp(-\frac{1}{2}\abs{\beta/w}^{2+\epsilon})$ and $\Omega_w^2(\beta)= \exp(-\frac{1}{2}\abs{\beta/w}^{2+\epsilon})$ and note that both $\Omega_w^1(\beta)$ and $\Omega_w^2(\beta)$ are square integrable functions.

For property (b), we note that $\Omega_w^1(\beta)e^{\pi^2\abs{\beta}^2/2} = \exp(-\frac{1}{2}\abs{\beta/w}^{2+\epsilon}) \exp(\pi^2\abs{\beta}^2/2)  \approx \exp(-\frac{1}{2}\abs{\beta/w}^{2+\epsilon}) \leq \exp(-\frac{1}{2}\abs{\beta/w}^{2})$ for sufficiently large $\abs{\beta} \gg1$. The last term is just a Gaussian function, which is square integrable, so $\Omega_w^1(\beta)e^{\pi^2\abs{\beta}^2/2}$ is also square integrable.

For property (c), one can verify that $\Omega_w(0)=e^0 = 1$ and that for any given $\beta$, as $w\rightarrow \infty$, $\Omega_{w, \epsilon}(\beta) \rightarrow 1$.

For property (d), one can verify that $\Omega_{w, \epsilon}(\beta) = \Omega_{w/\abs{r}, \epsilon}(\beta)\Omega_{t, \epsilon}(\beta) $ where $t= \frac{w}{(1-\abs{r}^{q})^{1/q}}$ where $q = 2+\epsilon$.

For property (e), one can verify that $\Omega_{kw, \epsilon}(k\beta)= \exp(-\abs{k\beta/(kw)}^{2+\epsilon})= \exp(-\abs{\beta/w}^{2+\epsilon}) = \Omega_{w, \epsilon}(\beta)$.

Finally, we observe that as $\epsilon \rightarrow 0$, $\exp(-\abs{\beta/w}^{2+\epsilon}) \approx \exp(-\abs{\beta/w}^{2})$. Since a Gaussian function's Fourier transform is also Gaussian, $\mathcal{F}\Omega_{w,\epsilon}$ approaches a positive distribution so $\mathcal{N}(\Omega_{w=1, \epsilon}) \rightarrow 0 $ as $\epsilon \rightarrow 0$.
\end{proof}

\section{Examples}

Here, we provide some numerical examples that illustrates our results for the negativity $\mathcal N$, the $s$-parametrized negativity $\mathcal{N_s}$ and the filtered negativity $\mathcal N_{\Omega, w}\coloneqq \mathcal{N}(P_{\Omega,w})$ using several prominent nonclassical states. We will use the filter $\Omega_{w,\epsilon}$ from Proposition \ref{prop::filt}. The error parameter $\epsilon$ is choosen to be $\epsilon = 0.21$ such that $2\delta = 2\mathcal N(\Omega_{w=1, \epsilon}) \approx 0.05$. From Theorem~\ref{thm::filtNegMono}), this means that the resulting filtered negativity $\mathcal N_{\Omega, w}$ is a linear optical monotone up to approximately a 5 percent error. Note that this choice is arbitrary, as $\delta$ can be made as small as desired by decreasing $\epsilon$. 

For highly nonclassical states such as Fock and squeezed-vacuum states $\mathcal{N}$ is infinitely large, which can be verified numerically via Definition~\ref{def::NegP}. One example of a nonclassical state with finite $\mathcal{N}$ is the single-photon-added thermal(SPAT) state, defined by $\rho_\mathrm{SPAT} = a^\dagger e^{-\beta \hbar \omega a^\dagger a}a / \Tr(e^{-\beta \hbar \omega a^\dagger a}a a^\dagger)$. Its characteristic function is $\chi_\mathrm{SPAT}(\beta) = [1-\pi^2(1+\bar n )\abs{\beta}^2] e^{-\pi^2\abs{\beta}^2/\bar{n}}$, and the corresponding $P$-function is $P_\mathrm{SPAT}(\alpha) = \frac{1+\bar n}{\pi \bar n^3}\qty(\abs{\alpha}^2-\frac{\bar n}{1+ \bar n})e^{-\abs{\alpha}^2/\bar n}$ \cite{Kiesel2008}. Figure~\ref{fig::SPAT}, illustrates how the the filtered negativity $\mathcal N_{\Omega, w}(\rho_\mathrm{SPAT})$ approaches $\mathcal N (\rho_\mathrm{SPAT})$ as $w \rightarrow \infty$, which comes directly from Definition~ \ref{def::NegP}. From Theorem~$\ref{thm::monotonicity}$, we know that the negativity $\mathcal N (\rho_\mathrm{SPAT})$ cannot be increased via  linear optical processes.

\begin{figure}[t]
    \centering
    \includegraphics[width=\linewidth]{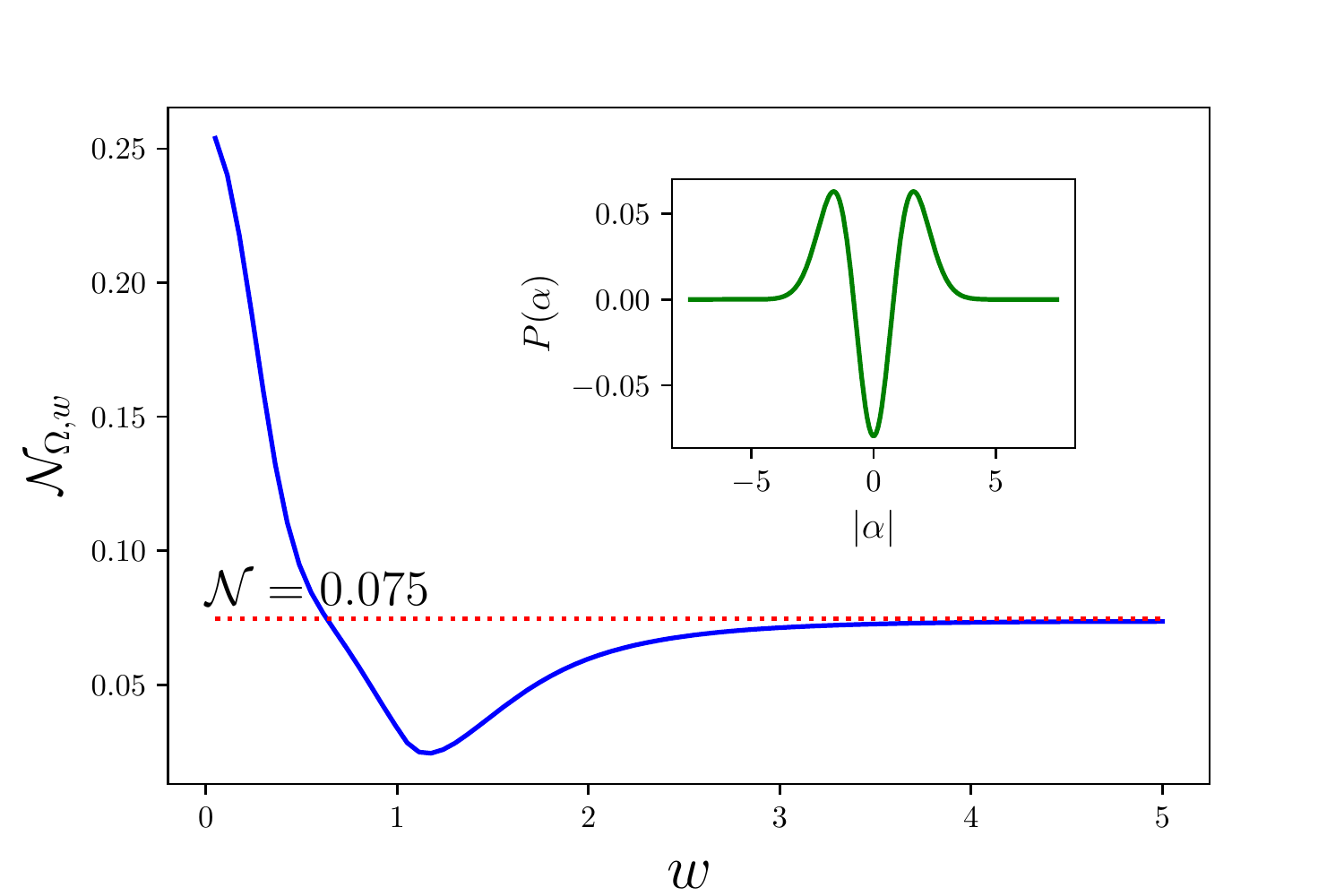}
    \caption{ Convergence of the filtered negativity(solid line) $\mathcal{N}_{\Omega,w}$ to the negativity(dotted line) $\mathcal{N}$ for the single photon added thermal state $\rho_\mathrm{SPAT}$ with $\bar n = 2$. }
    \label{fig::SPAT}
\end{figure}

From Theorem~\ref{thm::sMonotonicity} we know that the $s$-parametrized negativity $\mathcal{N}_s$ is a monotonically decreasing function of $s$. We illustrate this using Fock states $\ket{n}$. Its $s$-parametrized characteristic function is given by $\ket{n}$ is $\chi_{s}(\beta)=e^{(s-1)\pi^2\abs{\beta}/2} \mathrm{L_n}(\pi^2 \abs{\beta}^2)$, with the corresponding $s$-parametrized quasiprobabilities given by \cite{Wunsche1998} \begin{equation*}
    P_s(\alpha) = \frac{2}{\pi(1+s)}\qty(-\frac{1-s}{1+s})^n \exp(-\frac{2\abs{\beta}^2}{1+s}) \mathrm{L_n}\qty(\frac{4\abs{\beta}^2}{1-s^2}).
\end{equation*} 

Plotting $\mathcal{N}_s$, Figure \ref{fig::sNeg} illustrates its monotonic dependence on $s$ for $n =$ 1, 2 and 3. Also note how for every $s$, $\mathcal{N}_s(\ket{n})$ increases with $n$. Theorem~\ref{thm::sNegMono} says that $\mathcal{N}_s(\ket{n})$ for $s < 1$ are also valid, albeit weaker, nonclassicality measures, according to the resource theory of Refs.~\cite{Tan2017, Kwon2019}.

\begin{figure}[t]
    \centering
    \includegraphics[width=\linewidth]{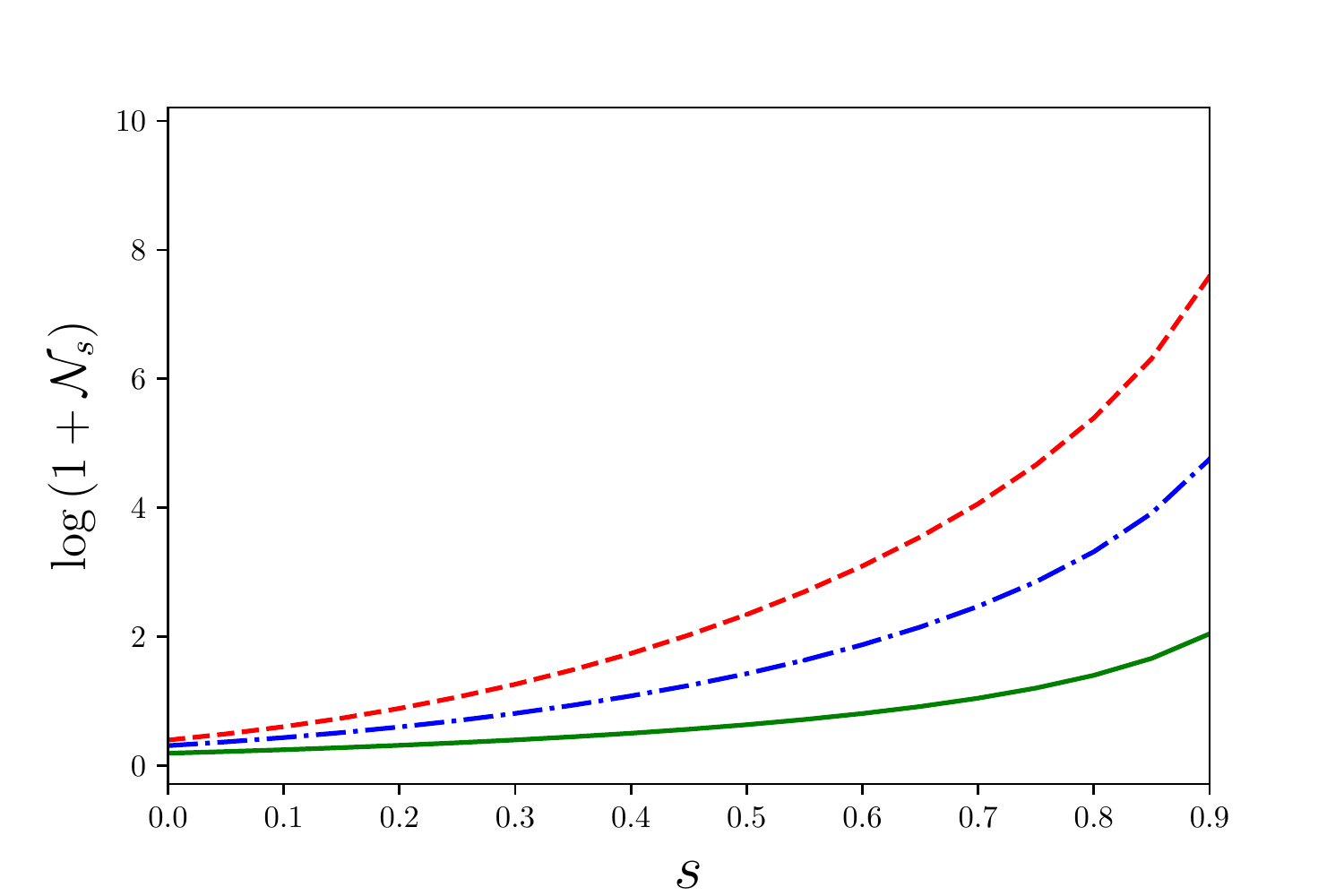}
    \caption{$s$-parametrized negativity of Fock state $\ket{n}$ for $n=$1(solid line), 2(dot-dashed line), and 3(dashed line).}
    \label{fig::sNeg}
\end{figure}

The $s$-parametrized negativities can be infinite in general. One example is the squeezed vacuum state $\ket{r} = e^{r({a^\dagger}^2-a^2)/2}\ket{0}$. Its characteristic function is $\chi_{\ket{r}}(\beta=x+iy) = \exp{\frac{\pi^2}{2}\qty[(s-e^{2r})x^2+(s-e^{-2r})y^2]}$ for $r>0$. If $s \leq e^{-2r}$, then the $s$-parametrized quasiprobability of $\ket{r}$ is Gaussian, so it does not show any negative value. However, if $s>e^{-2r}$, then its quasiprobability distribution shows extremely singular behavior, and one can numerically verify that $\mathcal{N}_s$ is infinite. In such cases, $\mathcal{N}_s$ is useful to identify the nonclassicality of the state, but is unable to capture the increase in nonclassicality that one gets from additional squeezing. This can be circumvented by considering the filtered negativity $\mathcal N_{\Omega, w}$. 

Figure~\ref{fig::fNeg} illustrates the filtered negativities $\mathcal N_{\Omega, w}$ of various squeezed states $\ket{r}$ and Fock states $\ket{n}$. We see that the filtered negativity captures the increase in nonclassicality due to both the increase in photon number $n$ and the increase in squeezing $r$. As the filter $\Omega_{w, \epsilon}$ has non-zero negativity, $\mathcal N_{\Omega, w}$ is only an approximate monotone (see Theorem~\ref{thm::filtNegMono}), but this error can be made arbitrarily small by decreasing the parameter $\epsilon$. This may, however, require increased numerical precision and hence additional computational costs.

\begin{figure}
    \centering

    \begin{minipage}{0.8\linewidth}
        \centering
        \includegraphics[width=\linewidth]{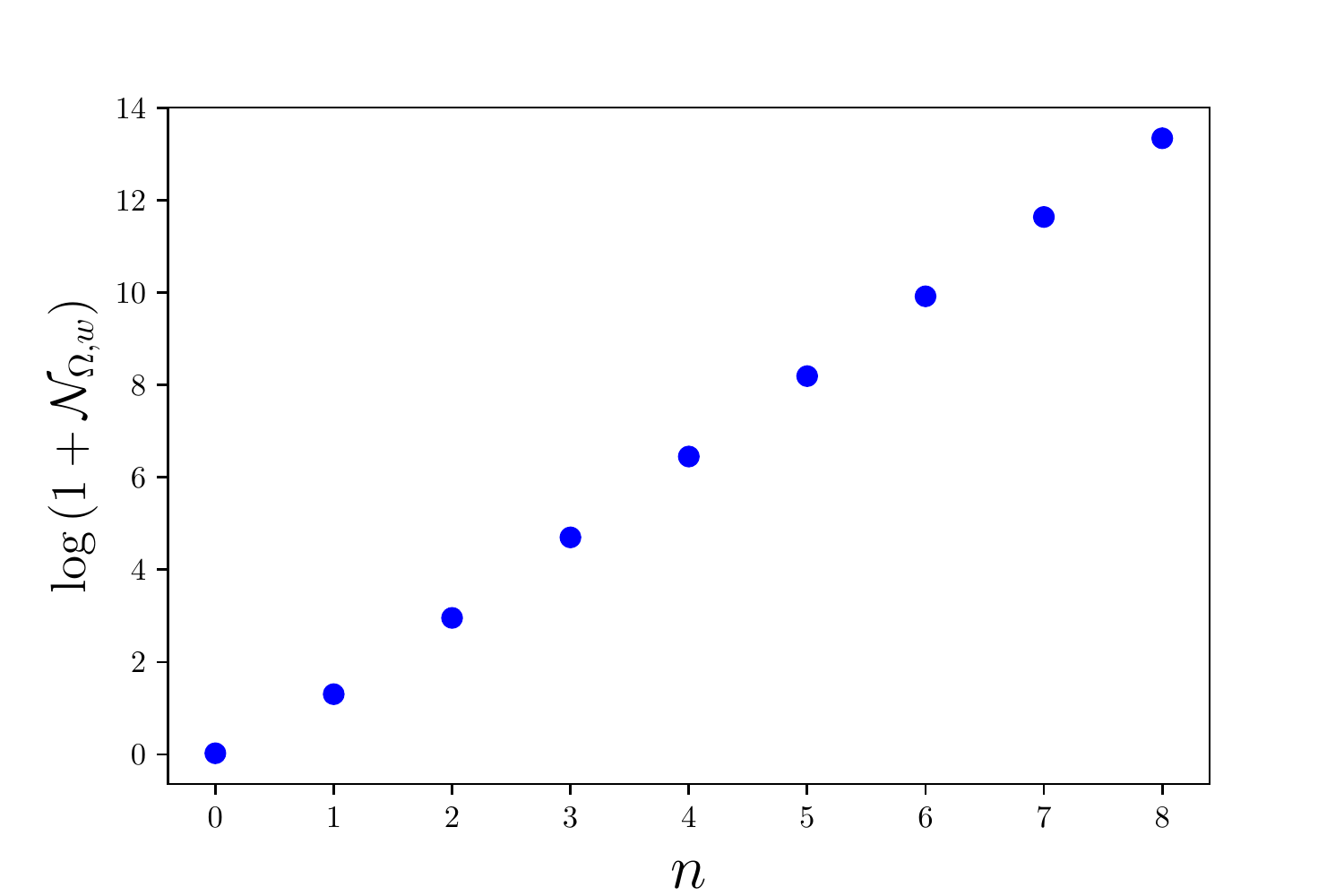} % first 
    \end{minipage}\hfill
    \begin{minipage}{0.8\linewidth}
        \centering
        \includegraphics[width=\linewidth]{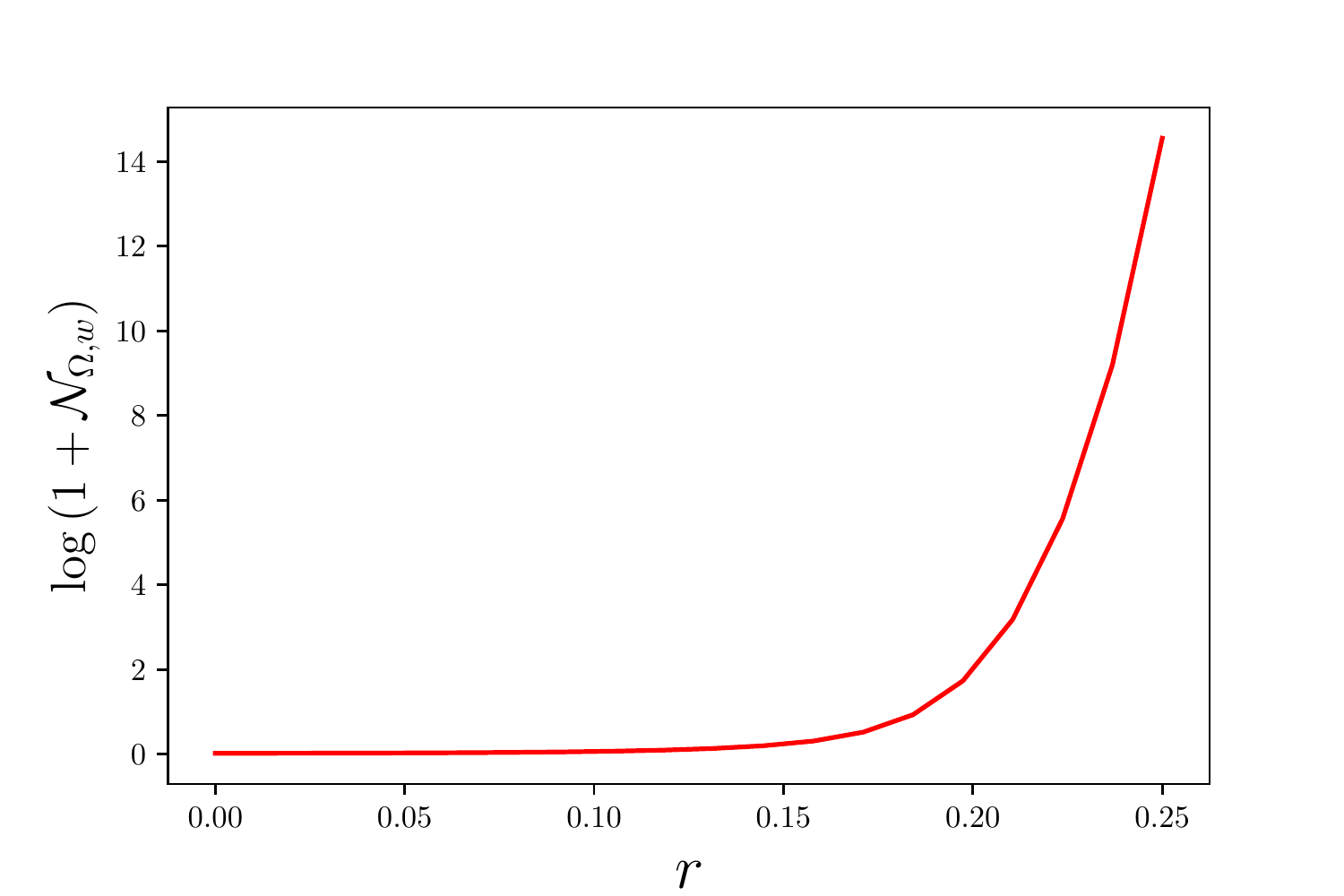} % second 
    \end{minipage}    
    \caption{The (logged) filtered negativity $\log(\mathcal{N}_{\Omega,w}+1)$ for Fock states $\ket{n}$ (top) and squeezed vacuum $\ket{r}$ (bottom). }
    
    \label{fig::fNeg}
\end{figure}

\section{Conclusion}

We introduced a method to unambiguously define the negativity of the $P$-function, and more generally, the negativity of the set of $s$-parametrized quasiprobabilities. Our method is based on a modified version of the filtered $P$-function in Ref.\cite{Kiesel2010}. Based on this definition, it is possible to show that negativity of the set of $s$-parametrized quasiprobabilities are all linear optical monotones, and form a continuous hierarchy of increasingly weaker nonclassicality measures that all belong to the operational resource theory of nonclassicality considered in Refs.\cite{Tan2017, Kwon2019}. 

In general, the $s$-parametrized negativities may have infinite values. In order to circumvent this, we introduce an approximate linear optical monotone that is computable and is able identify nearly every nonclassical state. A key advantage of this approach is that the set of unidentifiable nonclassical states can be made to converge to zero by increasing the single parameter $w$. The error can also be controlled via a single parameter $\epsilon$.

We also demonstrate in Theorem~\ref{thm::robustness} that the negativity of the $P$-function has a direct operational interpretation as the amount of statistical mixing with classical noise required to erase nonclassicality. Since $\mathcal{N}(\rho)$ is not always finite, this means that there are some states whose nonclassicality cannot be erased by simple statistical mixing. This is a characteristic it shares with quantum coherence, where simple mixing with an incoherent state cannot make the state classical in general\cite{Napoli2016}. One may also consider the amount of statistical mixing with nonclassical noise as a measure of nonclassicality, but at present, it is not clear how one may compute such a quantity. We leave this for future work.

%We also comment on the generalizability of our approach to multimode scenarios. While only single mode states were considered in this article,  it should be clear that a direct extension to multimode states is also possible. For $N$ modes, one may simply substitute the complex variables for a vector of $N$ complex variables, such that a general state and filter may written as $\rho = \int d^{2N} \alpha P(\vec{\alpha}) \ket{\vec{\alpha}}\bra{\vec{\alpha}}$ and $\Omega_w(\vec{\beta})$, where $\vec{\alpha} = (\alpha_1 , \ldots, \alpha_N)$ and $\vec{\beta} = (\beta_1 , \ldots, \beta_N)$. The essential definitions and properties will still follow.

Finally, we comment that our proposed measures are practical under realistic settings. In order to compute the proposed measures, one only requires the characteristic function of the quantum state, with no limitations on whether the state is mixed or pure.  The characteristic function may be sampled directly in the laboratory using only homodyne measurements\cite{Kiesel2009}. More generally, the reconstruction of any of the $s$-parametrized quasiprobabilities\cite{Lvovsky2009} allows you to infer the characteristic function, and hence compute our proposed measures.

We hope our work will spur continued interest in the study of nonclassicality in light fields.

\section{Acknowledgements}
This work was supported by the National Research Foundation of Korea (NRF) through a grant funded by the Korea government (MSIP) (Grant No. 2010-0018295). K.C. Tan was supported by Korea Research Fellowship Program through the National Research Foundation of Korea (NRF) funded by the Ministry of Science and ICT (Grant No. 2016H1D3A1938100). S. Choi was supported by NRF(National Research Foundation of Korea) Grant funded by Korean Government(NRF-2016H1A2A1908381-Global Ph.D. Fellowship Program).

\end{document}